\documentclass[11pt]{article}
\usepackage{a4wide}
\usepackage[usenames,dvipsnames]{color}
\usepackage{amssymb,amsmath,graphics,color,enumerate}
\usepackage{tikz}
\usepackage{xspace}
\usetikzlibrary{arrows,shapes,snakes,automata,backgrounds,petri,patterns} 
\usepackage{todonotes}
\usepackage[utf8x,utf8]{inputenc} 
\usepackage[T1]{fontenc}
\usepackage{booktabs}
\usepackage{array}
\usepackage{mathtools}
\usepackage{algorithm}
\usepackage[noend]{algpseudocode}
\usepackage[numbers,sort]{natbib}

\usepackage{standalone}
\usepackage{tikz}
\usetikzlibrary{arrows}
\usetikzlibrary{patterns}

\usepackage{epsfig}
\usepackage{subfig}
\usepackage{amsmath,amsfonts,amssymb,epsfig,color,amsthm}
\usepackage{comment}
\usepackage{thmtools}
\usepackage{thm-restate}
\usepackage{mdframed}
\usepackage[absolute]{textpos}

\newtheorem{theorem}{Theorem}[section]
\newtheorem{lemma}[theorem]{Lemma}
\newtheorem{corollary}[theorem]{Corollary}

\newtheorem{observation}[theorem]{Observation}

\newcommand{\Cc}{{\ensuremath{\mathcal{C}}}}

\newcommand{\Ff}{{\ensuremath{\mathcal{F}}}}

\newcommand{\cc}{\mathop{\rm{cc}}}
\newcommand{\poly}{\mathop{\rm{poly}}}
\newcommand{\Oh}{{\ensuremath{\mathcal{O}}}}
\newcommand{\Ohstar}{{\ensuremath{\mathcal{O}^*}}}

\newcommand{\PM}{{\ensuremath{\mathcal{PM}}}}

\newcommand{\Z}{\mathbb{Z}}
\newcommand{\Zq}{\mathbb{Z}_{\ge 0}}

\newcommand{\anonymyze}[1]{#1}%

\newcommand{\ceil}[1]{\ensuremath{\left\lceil{#1}\right\rceil}}%
\newcommand{\ignore}[1]{}%
\newcommand{\field}[1]{\textup{GF}(#1)}

\newcommand{\ProblemFormat}[1]{{\sc #1}}
\newcommand{\ProblemName}[1]{\ProblemFormat{#1}\xspace}

\newcommand{\In}{{\sf in}}
\newcommand{\Out}{{\sf out}}
\newcommand{\indeg}{{\sf indeg}}
\newcommand{\outdeg}{{\sf outdeg}}

\newcommand{\probDirHam}{\ProblemName{Directed Hamiltonicity}}
\newcommand{\probMVTSP}{\ProblemName{Many Visits TSP}}
\newcommand{\probFixDegFSub}{\ProblemName{Fixed Degree $\Ff$-Subgraph}}
\newcommand{\probFixDegConSub}{\ProblemName{Fixed Degree Connected Subgraph}}
\newcommand{\probDecUFixDegConSub}{\ProblemName{Decision Unweighted Fixed Degree Connected Subgraph}}

\newcommand{\probFixDegOutSub}{\ProblemName{Fixed Degree Subgraph With Outbranching}}
\newcommand{\probFixDegSub}{\ProblemName{Fixed Degree Subgraph}}
\newcommand{\heading}[1]{\medskip\noindent{\bf #1.\ }}%
\newcommand{\Ninfty}{\Zq \cup\{\infty\}}
\newcommand{\OPT}{{\rm OPT}}
\newcommand{\ALG}{{\rm ALG}}

\newcommand{\myroot}{{\sf root}}

\newcommand{\defproblem}[3]{
	\vspace{1mm plus 2mm minus 1mm}
	\noindent\fbox{\begin{minipage}{0.96\textwidth}
		#1\\
		\textbf{Input:} #2  \\
		\textbf{Question:} #3
	\end{minipage}}
  \vspace{1mm plus 2mm minus 1mm}
}

\begin{document}

\title{Many visits TSP revisited\anonymyze{\thanks{This research is a part of projects that have received funding from the European Research Council (ERC)
		under the European Union's Horizon 2020 research and innovation programme
		Grant Agreement 714704 (S.~Li, W. Nadara) and 677651 (Ł. Kowalik, M. Smulewicz).}} }

\date{\today}

\anonymyze{
\author{Łukasz Kowalik\thanks{Institute of Informatics, University of Warsaw, Poland (\texttt{kowalik@mimuw.edu.pl})} \and Shaohua Li\thanks{Institute of Informatics, University of Warsaw, Poland (\texttt{shaohua.li@mimuw.edu.pl})} \and Wojciech Nadara\thanks{Institute of Informatics, University of Warsaw, Poland (\texttt{w.nadara@mimuw.edu.pl})} \and Marcin Smulewicz\thanks{Institute of Informatics, University of Warsaw, Poland (\texttt{m.smulewicz@mimuw.edu.pl})} \and Magnus Wahlström\thanks{Royal Holloway, University of London, UK	(\texttt{Magnus.Wahlstrom@rhul.ac.uk})}}
}

\maketitle

\anonymyze{
\begin{textblock}{20}(0, 13.0)
	\includegraphics[width=40px]{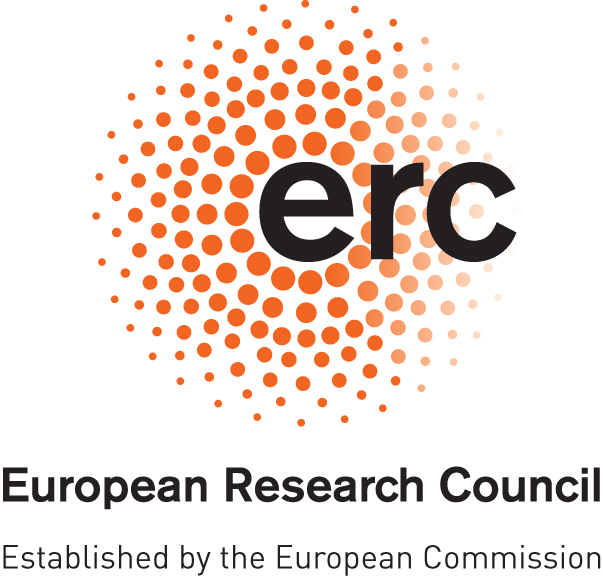}%
\end{textblock}
\begin{textblock}{20}(-0.25, 13.4)
	\includegraphics[width=60px]{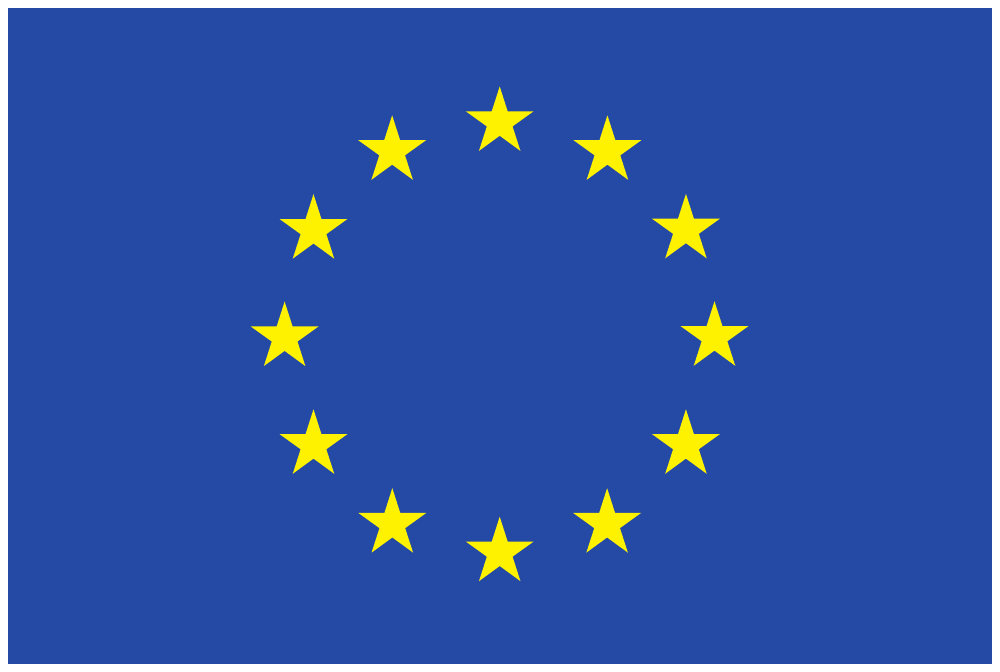}%
\end{textblock}
}

\begin{abstract}
We study the \probMVTSP problem, where given a number $k(v)$ for each of $n$ cities and pairwise (possibly asymmetric) integer distances, one has to find an optimal tour that visits each city $v$ exactly $k(v)$ times.
The currently fastest algorithm is due to Berger, Kozma, Mnich and Vincze [SODA 2019, TALG 2020] and runs in time and space $\Ohstar(5^n)$. 
They also show a polynomial space algorithm running in time $\Oh(16^{n+o(n)})$. In this work, we show three main results:
\begin{itemize}
	\item A randomized polynomial space algorithm in time $\Ohstar(2^nD)$, where $D$ is the maximum distance between two cities. By using standard methods, this results in $(1+\epsilon)$-approximation in time $\Ohstar(2^n\epsilon^{-1})$. 
	Improving the constant $2$ in these results would be a major breakthrough, as it would result in improving the $\Ohstar(2^n)$-time algorithm for {\sc Directed Hamiltonian Cycle}, which is a 50 years old open problem.
	\item A tight analysis of Berger et al.'s exponential space algorithm, resulting in $\Ohstar(4^n)$ running time bound.
	\item A new polynomial space algorithm, running in time $\Oh(7.88^n)$.
\end{itemize}
\end{abstract}

\newpage

\section{Introduction}
In the \probMVTSP (MVTSP) we are given a set $V$ of $n$ vertices, with pairwise distances (or costs) $d:V^2\rightarrow \Zq \cup \{\infty \}$.
We are also given a function $k:V\rightarrow\Z_{+}$.
A valid tour of length~$\ell$ is a sequence of vertices $(x_1, \dots, x_\ell)$, where $\ell=\sum_{v \in V}{k(v)}$, such that each $v \in V$ appears in the sequence exactly $k(v)$ times.
The cost of the tour is $\sum_{i=1}^{\ell-1} d(x_i,x_{i+1}) + d({x_\ell,x_1})$. Our goal is to find a valid tour with minimum cost. 

\probMVTSP is a natural generalization of the classical (asymmetric) {\sc Traveling Salesman Problem} (TSP), which corresponds to the case when $k(v)=1$ for every vertex $v$.
Similarly as its special case, MVTSP arises with a variety of applications, including scheduling~\cite{Psaraftis1980, HochbaumShamir1991,BraunerEtAl2005,vanderVeenZhang1996,flowshop}, computational geometry~\cite{scatterTSP} and parameterized complexity~\cite{Lampis10}.	

\subsection{Related work}
The standard dynamic programming for TSP of Bellman~\cite{BellmanTSP}, Held and Karp~\cite{HeldKarpTSP} running in time $\Ohstar(2^n)$ can be easily generalized to MVTSP resulting in an algorithm with the running time of $\Ohstar(\prod_{v \in V}{(k(v)+1)})$, as noted by Psaraftis~\cite{Psaraftis1980}.
A breakthrough came in the work of Cosmadakis and Papadimitriou~\cite{papadimitriou} who presented an algorithm running in time $2^{O(n\log n)}+\Oh(n^3\log\ell)$ and space $2^{O(n\log n)}$, thus essentially removing the dependence on the function $k$ from the bound (the $\log\ell$ factor can be actually skipped if we support the original algorithm with a today's state-of-the-art minimum cost flow algorithm).
This may be surprising since the {\em length} of the output sequence is $\ell$.
However, beginning from the work of Cosmadakis and Papadimitriou we consider MVTSP with compressed output, namely the output is a multiplicity function which encodes the number of times every edge is visited by the tour. By using a standard Eulerian tour algorithm we can compute an explicit tour from this output.

The crux of the approach of Cosmadakis and Papadimitriou~\cite{papadimitriou} was an observation that every solution can be decomposed to a minimal connected spanning Eulerian subgraph (which enforces connectivity of the solution) and subgraph satisfying appropriate degree constraints (which completes the tour so that the numbers of visits agree). Moreover, once we guess the degree sequence $\delta$ of the Eulerian subgraph, our task splits into two separate tasks: finding a cheapest minimal connected Eulerian subgraph consistent with $\delta$ (which is computationally hard) and finding a cheapest subgraph satisfying the degree  constraints  (which can by solved in polynomial time by a reduction to minimum cost flow).

Yet another breakthrough came only recently, namely Berger, Kozma, Mnich and Vincze~\cite{berger-soda,berger-arxiv} improved the running time to $\Ohstar(5^n)$.
Their main contribution is an idea that it is more convenient to use outbranchings (i.e. spanning trees oriented out of the root) to force connectivity of the solution. 
The result of Berger et al.\ is the first algorithm for MVTSP which is optimal assuming Exponential Time Hypothesis (ETH)~\cite{eth}, i.e., there is no algorithm in time $2^{o(n)}$, unless ETH fails.
Moreover, by applying the divide and conquer approach of Gurevich and Shelah~\cite{GurevichShelah} they design a polynomial space algorithm, running in time $\Oh(16^{n+o(n)})$.

\subsection{Our results}

In this work, we take the next step in exploration of the \probMVTSP problem: we aim at algorithms which are optimal at a more fine grained level, namely with running times of the form $\Oh(c^n)$, such that an improvement to $\Oh((c-\epsilon)^n)$ for any $\epsilon>0$ meets a kind of natural barrier, for example contradicts Strong Exponential Time Hypothesis (SETH)~\cite{seth} or Set Cover Conjecture (SCC)~\cite{secoco}. Our main result is the following theorem.

\begin{theorem}
	\label{thm:bnd-tsp}
	There is a randomized algorithm that solves \probMVTSP in time $\Ohstar(2^nD)$ and polynomial space, where $D=\max\{d(u,v) : u, v \in V, d(u, v) \neq \infty\}$.
    The algorithm returns a minimum weight solution with constant probability.
\end{theorem}

The natural barrier in this case is connected with \probDirHam, the problem of determining if a directed graph contains a Hamiltonian cycle. Indeed, this is a special case of \probMVTSP with $D=1$, so an improvement to $\Ohstar(1.99^nD)$ in Theorem~\ref{thm:bnd-tsp} would result in an algorithm in time $\Ohstar(1.99^n)$ for \probDirHam.
While it is not known whether such an algorithm contradicts SETH or SCC, the question about its existence is a major open problem which in the last 58 years has seen some progress only for special graph classes, like bipartite graphs~\cite{bkk:laplacians,cygan:bases}.

At the technical level, Theorem~\ref{thm:bnd-tsp} uses so-called algebraic approach and relies on two key insights. The first one is to enforce connectivity not by guessing a spanning connected subgraph as in the previous works, but by applying the Cut and Count approach of Cygan et al~\cite{cutcount}. The second insight is to satisfy the degree constraints using Tutte matrix~\cite{tutte,lovasz-tutte}.

By using standard rounding techniques, we are able to make the algorithm from Theorem~\ref{thm:bnd-tsp} somewhat useful even if the maximum distance $D$ is large.
Namely, we prove the following.

\begin{theorem}
	\label{thm:aprox}
	For any $\epsilon > 0$ there is a randomized $(1+\epsilon)$-approximation algorithm that solves \probMVTSP in $\Ohstar({2^n}{\epsilon^{-1}})$ time and polynomial space.
\end{theorem}

In Theorems~\ref{thm:bnd-tsp} and~\ref{thm:aprox} the better exponential dependence in the running time was achieved at the cost of sacrificing an $\Oh(D)$ factor in the running time, or the optimality of the solution. 
What if we do not want to sacrifice anything?
While we are not able to get a $\Ohstar(2^n)$ algorithm yet, we are able to report a progress compared to the algorithm of Berger et al. in time $\Ohstar(5^n)$.
In fact we do not show a new algorithm but we provide a refined analysis of the  previous one.
The new analysis is tight (up to a polynomial factor).

\begin{theorem}
	\label{thm:expspace}
	There is an algorithm that solves \probMVTSP in time and space $\Ohstar(4^n)$.
\end{theorem}

In short, Berger et al.'s polyspace $\Ohstar(16^{n+o(n)})$ time algorithm iterates through all $O(4^n)$ degree sequences of an outbranching, finds the cheapest outbranching for each sequence in time $O(4^{n+o(n)})$, and completes it to satisfy the degree constraints using a polynomial time flow computation. 
Note that it is hard to speed up the cheapest outbranching routine, because for the sequence of $n-1$ ones and one zero we get essentially the TSP, for which the best known polynomial space  algorithm takes time $O(4^{n+o(n)})$~\cite{GurevichShelah}.
However, we are still able get a significant speed up of their algorithm, roughly, by using a more powerful minimum cost flow network, which allows for computing the cheapest outbranchings in smaller subgraphs.

\begin{theorem}
	\label{thm:polyspace}
	There is an algorithm that solves \probMVTSP in time $\Ohstar(7.88^n)$ and polynomial space.
\end{theorem}



\heading{Organization of the paper}
In Section~\ref{sec:reduction} we show that, essentially, using a polynomial time preprocessing we can reduce an instance of \probMVTSP to an equivalent one but with demands $\In$, $\Out$ bounded $O(n^2)$. This reduction is a crucial prerequisite for Section~\ref{sec:2^n} where we prove Theorem~\ref{thm:bnd-tsp}. 
Next, in Section~\ref{sec:general} we prove Theorem~\ref{thm:expspace} and in Section~\ref{sec:polyspace} we prove Theorem~\ref{thm:polyspace}.
We note that in these two sections we do not need the reduction from Section~\ref{sec:reduction}, however, in practice, applying it, which should speed-up the flow computations used in both algorithms described there. 
Finally, in Section~\ref{sec:approx} we show Theorem~\ref{thm:aprox} and we discuss further research in Section~\ref{sec:further}.

\section{Preliminaries}

We use Iverson bracket, i.e., if $\alpha$ is a logical proposition, then the expression $[\alpha]$ evaluates to $1$ when $\alpha$ is true and $0$ otherwise.

For two integer-valued functions $f,g$ on the same domain $D$, we write $f\le g$ when $f(x)\le g(x)$ for every $x\in D$.
Similarly, $f+g$ (resp. $f-g$) denote the pointwise sum (difference) of $f$ and $g$.
This generalizes to functions on different domains $D_f$, $D_g$ by extending the functions to $D_f\cup D_g$ so that the values outside the original domain are $0$.

For a cost function $d:V^2\rightarrow\Ninfty$, and a multiplicty function $m:V^2\rightarrow\Zq$ we denote the cost of $m$ as $d(m)=\sum_{u,v\in V^2}d(u,v)m(u,v)$.

\heading{Multisets} 
Recall that a {\em multiset} $A$ can be identified by its {\em multiplicity }function $m_A:U\rightarrow\mathbb{Z}_{\ge 0}$, where $U$ is a set. 
We write $e\in A$ when $e\in U$ and $m_A(e)>0$. 
Consider two multisets $A$ and $B$. 
We write $A\subseteq B$ when for every $e\in A$ we have $e\in B$ and $m_A(e) \le m_B(e)$.
Also, $A=B$ when $A\subseteq B$ and $B\subseteq A$.
Assume w.l.o.g.\ that $m_A$ and $m_B$ have the same domain $U$.
Operations on multisets are defined by the corresponding multiplicites as follows: 
for every $e\in U$, we have 
$m_{A\cup B}(e)=\max\{m_A(e),m_B(e)\}$,
$m_{A\cap B}(e)=\min\{m_A(e),m_B(e)\}$,
$m_{A\setminus B}(e)=\max\{m_A(e)-m_B(e),0\}$,
$m_{A\bigtriangleup B}(e)=m_{(A\setminus B)\cup(B\setminus A)}=|m_A(e)-m_B(e)|$.
This notation extends to the situation when $A$ or $B$ is a set, by using the indicator function $m_A(e)=[e\in A]$.

\heading{Directed graphs}
Directed graphs (also called digraphs) in this paper can have multiple edges and multiple loops, so sets $E(G)$ will in fact be multisets.
We call a directed graph {\em simple} if it has no multiple edges or loops.
We call it {\em weakly simple} if it has no multiple edges or multiple loops (but single loops are allowed).
For a digraph $G$ by $G^\downarrow$ we denote the {\em support} of $G$, i.e., the weakly simple graph on the vertex set $V(G)$ such that $E(G^\downarrow)=\{(u,v) \mid \text{$G$ has an edge from $u$ to $v$}\}$.

Given a digraph $G=(V,E)$ we define its {\em multiplicity function} $m_G:V^2\rightarrow\Zq$
as the multiplicity function of its edge multiset, i.e., for any pair $u,v\in V$, we put $m_G(u,v)=m_{E}((u,v))$. 
Conversely, for a function $m:V^2\rightarrow\Zq$ we define the {\em thick graph} $G_m=(V,E)$ so that $m_G=m$. 
Abusing a notation slightly, we will identify $m$ and $G_m$, e.g., we can say that $m$ is strongly connected, contains a subgraph, etc.

We call a directed graph {\em connected} if the underlying undirected graph is connected.
Similarly, a {\em connected component} of a digraph $G$ is a subgraph of $G$ induced by a vertex set of a connected component of the underlying undirected graph.

For a graph $G$ (directed or undirected) and a subset $X\subseteq V(G)$, by $G[X]$ we denote the subgraph induced by $X$.

\heading{Solutions}
The following observation follows easily from known properties of Eulerian digraphs.
\begin{observation}
	\label{obs:solutions}
	\probMVTSP has a tour of cost $c$ if and only if there is a {\em multiplicity function} $m_G:V^2\rightarrow\Zq$ of cost $c$ such that $m$ contains a spanning connected subgraph.
\end{observation}
Thanks to Observation~\ref{obs:solutions}, in the remainder of this paper we refer to multiplicity functions as solutions of MVTSP (and some related problems which we are going to define).
By standard arguments, the multiplicity function can be transformed to a tour in time $\Oh(\ell)$.
Moreover, Grigoriev and Van de Klundert~\cite{Grigoriev2006} describe an algorithm which transforms it to a compressed representation of the tour in time $O(n^4 \log{\ell})$.

\heading{Out-trees}
An {\em out-tree} is the digraph obtained from a rooted tree by orienting all edges away of the root.
If an out-tree $T$ is a subgraph of a directed graph $G$ and additionally $T$ spans the whole vertex set $V(G)$ we call $T$ an {\em outbranching}.
The sequence $\{\outdeg_T(v)\}_{v\in V(T)}$ is called the {\em outdegree sequence} of $T$.
Consider a set of vertices $X\subseteq V$, $|X|\ge 2$.
\begin{lemma}[Berger et al.~\cite{berger-soda}, Lemma 2.4]
	\label{lem:out-sequence}
	A sequence of nonnegative integers $\{d_v\}_{v\in X}$ is an outdegree sequence of an out-tree spanning $X$ and rooted at $r\in X$ if and only if $(i)$ $d_r\ge 1$ and $(ii)$ $\sum_{v\in X}d_v = |X|-1$.
\end{lemma}
A sequence $\{d_v\}_{v\in X}$ that satisfies $(i)$ and $(ii)$ will be called an {\em out-tree sequence rooted at $r$}, or {\em outbranching sequence rooted at $r$} when additionally $X=V$.
A $\delta$-out-tree means any subtree spanning $X$ with outdegree sequence $\delta$.

\section{Reduction to small demands}
\label{sec:reduction}

The goal of this section is to show that, essentially, using a polynomial time preprocessing we can reduce an instance of \probMVTSP to an equivalent one but with demands $\In$, $\Out$ bounded $O(n^2)$.

Consider the following problem, for a family of simple digraphs $\Ff$.

\defproblem{\probFixDegFSub}{$d:V^2\rightarrow\Zq \cup \{\infty\}$, $\In, \Out: V \rightarrow \Zq$}{Find a function $m:V^2\rightarrow \Zq$ such that
\begin{enumerate}[$(i)$]
\item $G_m$ contains a member of $\Ff$ as a spanning subgraph,
\item for every $v\in V$ we have $\In(v)=\indeg_{G_m}(v)$ and $\Out(v)=\outdeg_{G_m}(v)$, and
\end{enumerate}
so as to minimize the value of $d(m)=\sum_{v,w\in V}d(v,w)m(v,w)$.
}

In this paper, we will consider two versions of the problem: when $\Ff$ is the family of all oriented trees, called \probFixDegConSub, and when $\Ff$ is the family of all out-trees with a fixed root $r$, called \probFixDegOutSub. 
The role of $\Ff$ is to force connectivity of the instance.
Other choices for $\Ff$ can also be interesing, for example Cosmadakis and Papadimitriou~\cite{papadimitriou} consider the family of minimal Eulerian digraphs.

When considering the instance of \probFixDegFSub we will use the notation $n=|V|$ and $\ell=\sum_{v\in V}\In(v)$. (Clearly, we can assume that also $\ell=\sum_{v\in V}\Out(v)$, for otherwise there is no solution.)

Observe that if the image of $d$ is $\{0,+\infty\}$ we get the natural unweighted version, where we are given a graph with edge set $d^{-1}(0)$ and the goal is to decide if one can choose multiplicities of the edges so that the resulting digraph contains a member of $\Ff$ and its in- and outdegrees match the demands of $\In$ and $\Out$. 

The following observation follows by standard properties of Eulerian cycles in digraphs and the fact that every strongly connected graph contains an outbranching rooted at arbitrary vertex.

\begin{observation}
\label{obs:red}
\probMVTSP is a special case of both \probFixDegConSub and \probFixDegOutSub with $\In(v)=\Out(v)=k(v)$ for every vertex $v\in V$.
\end{observation}

In the following lemma, we consider the relaxed problem \probFixDegSub, defined exactly as \probFixDegFSub, but dropping the constraint that solutions must contain a member of $\Ff$. In what follows, $s_n(\Ff)=\max_{\substack{G\in \Ff, |V(G)|= n}}|E(G)|$.

\begin{lemma}
\label{lem:diff}
Fix an input instance $d:V^2\rightarrow\Ninfty$, $\In,\Out:V^2\rightarrow \Zq$.
For every optimal solution $r$ of \probFixDegSub there is an optimal solution $c'$ of \probFixDegFSub such that for every $u,v\in V$
\[|r(u,v) - c'(u,v)|\le s_{|V|}(\Ff).\]
\end{lemma}

\begin{proof}
Let $c$ be an arbitrary optimal solution of \probFixDegFSub and let $B$ be an arbitrary graph from $\Ff$ which is a spanning subgraph of $G_c$.
Our plan is to build an optimal solution $c'$ of \probFixDegFSub which contains $B$ and does not differ too much from $r$.


Define multisets $A_c = E(G_{c}) \setminus E(G_{r})$, $A_r = E(G_{r}) \setminus E(G_{c})$ and $A=A_c\cup A_r= E(G_{c}) \bigtriangleup E(G_{r})$. In what follows, by an {\em alternating cycle} we mean an even cardinality set of edges \[\{(v_0,v_1),(v_2,v_1),(v_2,v_3),(v_4,v_3)\ldots,(v_{2\ell-2},v_{2\ell-1}),(v_0,v_{2\ell-1})\},\]
where edges come alternately from $A_c$ and $A_r$. Note that an alternating cycle is not really a directed cycle, it is just an orientation of a simple undirected cycle.

Note that for every vertex $v\in V$, among the edges in $A$ that enter (resp.\ leave) $v$ the number of edges from $A_c$ is the same as the number of edges from $A_r$ (counted with corresponding multiplicities), since both $c$ and $r$ satisfy the degree constraints for the same instance.
It follows that $A$ can be decomposed into a multiset $\Cc$ of alternating simple cycles, i.e., 
\[m_A=\sum_{C\in\Cc}m_C,\]
where $m_C:V^2\rightarrow\mathbb{Z}_{\ge 0}$ and for each pair $u,v\in V$ we have $m_C(u,v)=[(u,v)\in C]\cdot m_\Cc(C)$. 
To clarify, we note that the sum above is over all cycles in $\Cc$, and not over all copies of cycles.

Denote $B^+=E(B)\setminus E(G_{r})$.
Since $B^+\subseteq A_c$, for each $e\in B^+$, there is at least one cycle in $\Cc$ that contains $e$. We choose an arbitrary such cycle and we denote it by $C_e$. (Note that it may happen that $C_e=C_{e'}$ for two different edges $e,e'\in B^+$.)
Let $\Cc^+=\{C_e \mid e\in B^+\}$.
Then we define $c'$, by putting for every $u,v\in V$

\begin{equation}
\label{eq:c'}
c'(u,v)=r(u,v) + (-1)^{[(u,v)\in A_r]}\sum_{C\in\Cc^+}[(u,v)\in C].
\end{equation}

In other words, $c'$ is obtained from $r$ by iterating over all cycles in $C\in \Cc^+$, and adding one copy of each edge of $C\cap A_c$ and removing one copy of each edge of $C\cap A_r$.

Let us show that $G_{c'}$ contains $B$. 
This is trivial for every $e\in B^+$. 
When $e\in E(B) \cap E(G_{r})$, consider two cases.
If $e\not\in A_r$, then $c'(e)\ge r(e)$, so $e\in G_{c'}$.
If $e\in A_r$, $m_A(e)=r(e)-c(e)$. Then $c'(e)= r(e)-\sum_{C\in\Cc^+}[(u,v)\in C]\ge r(e)-m_A(e)=c(e)\ge 1$, where the last inequality follows since $B\subseteq G_c$.

To see that $c'$ satisfies the degree constraints, recall that $r$ does so, and note that if in~\eqref{eq:c'} we consider only the summands corresponding to a single cycle $C\in\Cc^+$, then for every vertex we either add one outgoing edge and remove one outgoing edge, or add one ingoing edge and remove one ingoing edge, or we do not change the set of edges incident to it. 

For a cycle $C\in\Cc$ let $\delta(C)=d(A_c \cap C)-d(A_r \cap C)$.
Observe that for every cycle $C\in\Cc$ we have $\delta(C)\ge 0$, for otherwise $E(G_{r})\setminus (C\cap A_r) \cup (C\cap A_c)$ contradicts the optimality of $r$. It follows that 
\begin{equation}
d(c')=d(r)+\sum_{C\in\Cc^+ } \delta(C) \le d(r)+\sum_{C\in\Cc} \delta(C) = d(c).
\end{equation}
Hence, since $c$ is optimal solution of \probFixDegFSub, we get that $c'$ is optimal solution of \probFixDegFSub as well. Moreover, by~\eqref{eq:c'}, for every $u,v\in V$, 
\begin{equation}
|c'(u,v)-r(u,v)| \le |\Cc^+| \le |B| \le s_{|V|}(\Ff).
\end{equation}
This ends the proof.
\end{proof}


As noted in~\cite{papadimitriou,berger-arxiv}, \probFixDegSub can be solved by a reduction to minimum cost flow. By applying Orlin's algorithm~\cite{Orlin93} we get the following.

\begin{observation}[Folklore,~\cite{papadimitriou,berger-arxiv}]
	\label{obs:flows}
\probFixDegSub can be solved in time $O(n^3 \log n)$.
\end{observation}

\begin{theorem}[Kernelization]	
\label{thm:kernel}
There is a polynomial time algorithm which, given an instance $I=(d,\In,\Out)$ of \probFixDegFSub outputs an instance $I'=(d,\In',\Out')$ of the same problem and a function $f:V^2\rightarrow\Zq$ such that 
\begin{enumerate}[$(i)$]
\item $\In'(v),\Out'(v)= \Oh(n\cdot s_n(\Ff))$ for every vertex $v$,
\item if $m^*$ is an optimal solution for $I'$, then $f+m^*$ is an optimal solution for $I$.
\end{enumerate}
The algorithm does need to know $\Ff$, just the value of $s_n(\Ff)$.
\end{theorem}

\begin{proof}
Our algorithm begins by finding an optimal solution $r$ of \probFixDegSub using Observation~\ref{obs:flows}.

Define $f_0:V^2\rightarrow \Zq$, where for every $v,w\in V$ we put 
$f_0(v,w)=\max\{r(v,w)-s_n(\Ff),0\}.$ By Lemma~\ref{lem:diff}, there exists an optimal solution $c'$ for instance $I$ such that $c'\ge f_0$.
Now define $f:V^2\rightarrow \Zq$, where for every $v,w\in V$ we put $f(v,w)=\max\{f_0(v,w)-1,0\}$.
Finally, we put $\In'(v)=\In(v)-\sum_{w\in V}f(w,v)$ and $\Out'(v)=\Out(v)-\sum_{w\in V}f(v,w)$.
The algorithm outputs $I'=(d,\In',\Out')$ and $f$.
In what follows, we show that the output has the desired properties.

For the property $(i)$, consider any vertex $v\in V$ and observe that $\sum_{w\in V}f(v,w)\ge \sum_{w\in V}(f_0(v,w)-1) \ge \sum_{w\in V}(r(v,w)-s_n(\Ff)-1)$.
Since $r$ is a feasible solution of $I$, we have $\Out(v)=\sum_{w\in V}r(v,w)$.
It follows that $\Out'(v)\le n(1+s_n(\Ff)) = O(n\cdot s_n(\Ff))$ as required. 
The argument for $\In'(v)$ is symmetric.


Now we focus on $(ii)$.
Let $m^*$ be an optimal solution for $I'$.
It is easy to check that $f+m^*$ satisfies the degree constraints for the instance $I$.
Also, since $m^*$ contains a subgraph from $\Ff$, then $f+m^*$ contains the same subgraph.
It follows that $f+m^*$ is a feasible solution of $I$.
It suffices to show that $f+m^*$ is an optimal solution for $I$.

Denote $r=c'-f$.
Consider any pair $v,w\in V$ such that $c'(v,w)\ge 1$. We claim that $f(v,w)\le c'(v,w)-1$. 
Indeed, if $f_0(v,w)=0$ then $f(v,w)=0\le c'(v,w)-1$, and if $f_0(v,w)\ge 1$ then $f(v,w)=f_0(v,w)-1\le c'(v,w)-1$.
It follows that $r(v,w)\ge 1$. 
In particular, since $c'$ contains a subgraph from $\Ff$, then also $r$ contains the same subgraph.
It follows that $r$ is a feasible solution for $I'$ (the degree constraints are easy to check).
Hence, $d(m^*) \le d(r)$.
It follows that $d(f+m^*)\le d(r+f)=d(c')$, so $f+m^*$ is indeed an optimal solution for $I$.
\end{proof}

\section{The small costs case in time $\Ohstar(2^nD)$}
\label{sec:2^n}

In this section we establish Theorem~\ref{thm:bnd-tsp}.
We do it in a bottom-up fashion, starting with a simplified core problem, and next generalizing the solution in a few steps.

\subsection{Unweighted decision version with small degree demands}
\label{sec:bnd-decision}

Consider the following problem.

\defproblem{\probDecUFixDegConSub}{a digraph $G=(V,E)$, $\In, \Out: V \rightarrow \Zq$}{Is there a function $m:V^2\rightarrow\Zq$ such that $G^\downarrow_m$ is a connected subgraph of $G$ and for every $v\in V$ we have $\In(v)=\indeg_{G_m}(v)$ and $\Out(v)=\outdeg_{G_m}(v)$?
}

Note that \probDecUFixDegConSub generalizes the directed Hamiltonian cycle problem, which is known to be solvable in $\Ohstar(2^n)$ time and polynomial space. In this section we show that this running time can be obtained for the more general problem as well, though we need to allow some randomization.

\begin{theorem}
	\label{thm:bnd-decision}
	There is a randomized algorithm which solves an instance $I=(\In,\Out)$ of \probDecUFixDegConSub in time $\Ohstar(2^n\poly(M))$ and polynomial space, where $M=\max_{v}\max\{\In(v),\Out(v)\}$. 
	The algorithm is Monte Carlo with one-sided error, i.e., the positive answer is always correct and the negative answer is correct with probability at least $p$, for any constant $p<1$.
\end{theorem}

Our strategy will be to reduce our problem to detecting a perfect matching in a bipartite graph with an additional connectivity constraint. 

We define a bipartite graph $B_G=(O,I,E(B_G))$ as follows. 
Let $I=\{v^I_1,\ldots,v^I_{\In(v)} \mid v\in V(H)\}$, $O=\{v^O_1,\ldots,v^O_{\Out(v)} \mid v\in V(H)\}$, and
$E(B_G)=\{u^O_iv^I_j\mid (u,v)\in E(G)\}$.

\begin{observation}
	$|I|=|O|=\Oh(nM)$ and $|E(B_G)|\le E(G)M^2 = \Oh(n^2M^2)$.
\end{observation}

For an undirected graph $H$ by $\PM(H)$ we denote the set of perfect matchings in $H$.
We say that a matching $M$ in $B_G$ is {\em connected} when for every cut $(X,V\setminus X)$ with $\emptyset\ne X\subsetneq V$ the matching $M$ contains an edge $u^O_iv^I_j$ such that $u\in X$ and $v\in V\setminus X$ or $v\in X$ and $u\in V\setminus X$.

For a matching $M$ in $B_G$ we define a {\em contraction} of $M$ as function $m:V^2\rightarrow \Zq$ such that 
$m(u,v)=|\{u^O_iv^I_j \in M \mid i\in[\Out(u)], j\in[\In(v)]\}|$.
In other words $G_{m}$ is obtained from $M$ by (1) orienting every edge from $O$ to $I$ and (2) identyfing all vertices in $\{v^I_1,\ldots,v^I_{\In(v)}\}\cup\{v^O_1,\ldots,v^O_{\Out(v)}\}$ for every $v\in V$, and keeping the multiple edges and loops. 

\begin{lemma}
	\label{lem:red-to-matching}
	$(G,\In,\Out)$ is a yes-instance of \probDecUFixDegConSub iff graph $B_G$ contains a connected perfect matching.
\end{lemma}

\begin{proof}
    Let $M$ be a connected  perfect matching in $B_G$ and let $m$ be its contraction.
    We claim that $m$ is a solution of $(G,\In,\Out)$.
    By the definition of $B_G$, $G^\downarrow_{m}$ is a subgraph of $G$.
	Since $M$ is connected, $G_{m}$ is connected as well, and so is $G^\downarrow_{m}$. 
	Moreover, since $M$ is a perfect matching $\In(v)=\indeg_{G_m}(v)$ and $\Out(v)=\outdeg_{G_m}(v)$ for every vertex $v$. 

    For the other direction, let $m$ be a solution for $(G,\In,\Out)$.
    For every $v\in V$, there are exactly $\Out(v)$ edges leaving $v$ in $G_m$.
    Let us denote them $e^O_{v,1},\ldots,e^O_{v,\Out(v)}$.
    Similarly, let us denote all the edges entering $v$ by $e^I_{v,1},\ldots,e^I_{v,\In(v)}$.
    Then we define $M$ as the set of edges of the form $u^O_iv^I_j$ such that $G_m$ contains an edge $e=e^O_{u,i}=e^I_{v,j}$.
    The fact that $M$ is a perfect matching is clear from the construction.
    Also, $M$ is connected, for otherwise $G_m$ is not connected. 
\end{proof}

From now on, let $B=(O,I,E(B))$ be an arbitrary subgraph of $B_G$.
Define the following multivariate polynomial over $\field{2^t}$, for an integer $t$ to be specified later. 

\begin{equation}
R = \sum_{\substack{M\in\PM(B) \\ \text{$M$ is connected}}}\prod_{e\in M}x_e
\end{equation}

\begin{lemma}
	\label{lem:polynom-equiv}
	$R$ is not the zero polynomial if and only if $B$ contains a connected perfect matching.
\end{lemma}

\begin{proof}
	It is clear that if $R$ is non-zero then $B$ contains a connected perfect matching.
	For the reverse implication it suffices to notice that every summand in $R$ has a different set of variables, so it does not cancel out with other summands over $\field{2^t}$.
\end{proof}

Our strategy is to test whether $R$ is non-zero by means of DeMillo--Lipton--Schwartz--Zippel Lemma, which we recall below.

\begin{lemma}[DeMillo and Lipton~\cite{DeMilloLipton1978}, Schwartz~\cite{schwartz}, Zippel~\cite{zippel}]
	\label{lem:zs}
	Let $P(x_1, x_2, \ldots, x_m)$ be a nonzero polynomial of degree at most $d$ over a field $\mathbb{F}$ and let $S$ be a finite subset of $\mathbb{F}$.
	Then, the probability that $P$ evaluates to zero on a random element $(a_1, a_2, \ldots , a_m)\in S^m$ is bounded by $d/|S|$.
\end{lemma}

By Lemmas~\ref{lem:polynom-equiv} and \ref{lem:zs}, the task reduces to {\em evaluating} $R$ fast. To this end, we will define a different polynomial $P$ which is easier to evaluate and turns out to be equal to $R$ over $\field{2^t}$.

Consider a subset $X\subseteq V$.
Let $I_X=\{v^I_i \in I \mid v\in X, i=1,\ldots,\In(v)\}$ and $O_X=\{v^O_i \in O \mid v\in X, i=1,\ldots,\Out(v)\}$.
Abusing the notation slightly, we will denote $B[X]=B[I_X\cup O_X]$.
Define the following polynomial.

\begin{equation}
P_X = \sum_{M\in\PM(B[X])} \prod_{e\in M}x_e
\end{equation}

In what follows, $v^*$ is an arbitrary but fixed vertex of $V$. Define yet another polynomial.

\begin{equation}
P = \sum_{\substack{X\subseteq V\\v^*\in X}} P_X P_{V\setminus X}.
\end{equation}

\begin{lemma}
	$P=R$.
\end{lemma}

\begin{proof}
	For a matching $M$ in $B$ we say that a set $X\subseteq V$ is {\em consistent}	with $M$
	when $M$ does not contain an edge $u^O_iv^I_j$ such that $u\in X$ and $v\in V\setminus X$ or $v\in X$ and $u\in V\setminus X$.
	The family of all subsets of $V$ that are consistent with $M$ will be denoted by $\Cc(M)$. Then we can rewrite $P$ as follows.

	\[
	\begin{aligned}
	P=& 
	\sum_{\substack{X\subseteq V\\v^*\in X}} 
	\sum_{M_1\in\PM(B[X])}
	\sum_{M_2\in\PM(B[V\setminus X])}
    \prod_{e\in M_1\cup M_2} x_e
	& \text{[definition]}\\
	=&
	\sum_{M\in\PM(B)}
	\sum_{\substack{X\in\Cc(M)\\v^*\in X}} 
	\prod_{e\in M} x_e
	& \text{[group by $M=M_1\uplus M_2$]}\\
	=&
	\sum_{M\in\PM(B)}
	|\{X\in\Cc(M)\mid v^*\in X\}|
	\prod_{e\in M} x_e
	& \text{[trivial]}
	\end{aligned}
	\]
	
	Let us consider a perfect matching $M\in\PM(B)$ and the corresponding contraction $m$.
	Observe that the number of sets that are consistent with $M$ and contain a vertex $v^*$ is equal to $2^{\cc(M)-1}$, where $\cc(M)$ is the number of connected components of $G_m$. 
	Indeed, when $X$ is consistent with $M$, then for every connected component $Q$ of $G_m$, either $V(Q)\subseteq X$ or $V(Q)\subseteq V\setminus X$. For the component that contains $v^*$ the choice is fixed, while every choice for the remaining components defines a set consistent with $M$. It follows that when $M$ is not connected $\cc(M)\ge 2$, and the value of $2^{\cc(M)-1}$ is equal to $0$ in $\field{2^t}$, so the corresponding summand vanishes. On the other hand, if $M$ is connected, the corresponding summand equals just $\prod_{e\in M} x_e$ and it does not cancel out with another summand because the monomial has a unique set of variables. It follows that $P=R$.
\end{proof}

\begin{lemma}[Tutte, Lovasz~\cite{tutte,lovasz-tutte}]
	\label{lem:tutte}
	For an arbitrary set $X\subseteq V$, the polynomial $P_X$ can be evaluated using $\poly(n+M)$ field operations.
\end{lemma}

\begin{proof}
	Compute the determinant of the corresponding Tutte matrix of dimension $|O|\times|I|$. 	
\end{proof}

Let us now fix our field, namely $t=\lceil 1+\log n+\log M\rceil$.
Since arithmetic operations in $\field{2^t}$ can be performed in time $\Oh(t\log^2t)=\Oh(\log(n+M)\log^2\log(n+m))$,
by the definition of $P$ and Lemma~\ref{lem:tutte} we get the following corollary.

\begin{corollary}
	\label{cor:poly-P-eval}
	$P$ can be evaluated in time $2^n\poly(n+M)$.
\end{corollary}

\begin{lemma}
\label{lem:matching-decision}
There is a randomized algorithm which decides if $B$ contains a connected perfect matching in time $\Ohstar(2^n\poly(M))$ and polynomial space, where $M=\max_{v}\max\{\In(v),\Out(v)\}$. 
The algorithm is Monte Carlo with one-sided error, i.e., the positive answer is always correct and the negative answer is correct with probability at least $p$, for any constant $p<1$.
\end{lemma}

\begin{proof}
The algorithm evaluates polynomial $P$ using Corollary~\ref{cor:poly-P-eval} substituting a random element of $\field{2^t}$ for each variable, and reports `yes' when the evaluation is nonzero and `no' otherwise.
If it reported 'yes', then $P$ was a non-zero polynomial and by Lemma~\ref{lem:polynom-equiv} the answer is correct.
Assume it reported 'no' for a yes-instance.
By Lemma~\ref{lem:polynom-equiv} $P$ is non-zero.
Since $\deg P = |I| \le nM$, by Lemmma~\ref{lem:zs} the probability that $P$ evaluated to $0$ is bounded by $\deg P / 2^t \le 1/2$ and we can make this probability arbitrarily small by repeating the whole algorithm a number of times, and reporting `yes' if at least one evaluation was nonzero.
The claim follows. 
\end{proof}

Theorem~\ref{thm:bnd-decision} follows immediately from Lemma~\ref{lem:red-to-matching} and Lemma~\ref{lem:matching-decision} applied to $B_G$.

\subsection{Finding the solution}

\begin{lemma}
	\label{lem:finding}
	There is a randomized algorithm which, given a yes-instance of 
	\probDecUFixDegConSub,
	always returns the corresponding solution $m$ in expected time $\Ohstar(2^n\poly(M))$. 
	The time can be made deterministic at the cost of introducing arbitrarily small probability of failure.
\end{lemma}

In order to prove Lemma~\ref{lem:finding} we cast the problem in the setting of {\em inclusion oracles} from the work of Bj\"orklund et al.~\cite{syphilis}.
Consider a universe $U$ and an (unknown) family of {\em witnesses} $\mathcal{F}\subseteq 2^U$.
An {\em inclusion oracle} is a procedure which, given a query set $Y\subseteq U$, answers (either YES or NO) 
whether there exists at least one witness $W\in\mathcal{F}$ such that 
$W\subseteq Y$. Bj\"orklund et al. prove the following.

\begin{theorem}[\cite{syphilis}]
	\label{thm:syphilis}
	There exists an algorithm that extracts a witness of size $k$ in $\mathcal{F}$ 
	using in expectation at most 
	$O(k\log |U|)$ queries to a randomized inclusion oracle that has no 
	false positives but may output a false negative with probability at most 
	$p\leq \frac{1}4$. 
\end{theorem}

\begin{proof}[Proof of Lemma~\ref{lem:finding}]
	Let $U=E(B_G)$ and let $\Ff$ be the family of all connected perfect matchings in $B_G$. 
	Note that $|U|=O(n^2M^2)$ and witnesses in $\Ff$ have all size $|I|=O(nM)$.
	Then, Lemma~\ref{lem:matching-decision} provides a randomized inclusion oracle and we can apply Theorem~\ref{thm:syphilis}.
	(If one insists on deterministic, and not expected, running time, it suffices to chose a sufficiently large constant $r$ and stop the algorithm if it exceeds the expected running time at least $r$ times --- by Markov's inequality, this happens with probability at most $1/r$.)
\end{proof}

\subsection{Proof of Theorem~\ref{thm:bnd-tsp}}

In the lemma below we will adapt the construction from Section~\ref{sec:bnd-decision} to the weighted case in a standard way, by introducing a new variable tracking the weight.

\begin{lemma}
	\label{lem:bnd-tsp-small-demands}
	There is a randomized algorithm which solves an instance $I=(d,\In,\Out,w)$ of \probFixDegConSub in time $\Ohstar(2^nD\poly(M))$ and polynomial space, where $M=\max_{v}\max\{\In(v),\Out(v)\}$ and $D$ is the maximum integer value of $d$. 
	The algorithm returns a minimum weight solution with probability at least $p$, for any constant $p<1$.
\end{lemma}

\begin{proof}
	Define $G=(V,E)$ where $E=\{(u,v)\in V^2\mid d(u,v)\in\Zq\}$.
	Let $R'$ be the polynomial obtained from $R$ by replacing every variable $x_e$ for $e=u^O_iv^I_j\in E(B_G)$ by the product $x_e\cdot y^{d(u,v)}$, where $y$ is a new variable.
	Proceed similarly with $P$, obtaining $P'$. 
	By Lemma~\ref{lem:polynom-equiv}, $P'=R'$.
	Decompose $R'$ as $R'=\sum_{i=0}^{|I|\cdot D}R'_iy^i$, where $R'_i$, for every $i=0,\ldots,|I|\cdot D$, is a polynomial in variables $\{x_e\}_{e\in E(B_G)}$. 
	The monomials in $R'_i$ enumerate all matchings $M$ such that the contraction $m$ of $M$ has weight $d(m)=i$. By the construction in the proof of Lemma~\ref{lem:red-to-matching} $R'_i$ is non-zero if and only if instance $I$ has a solution of weight $i$.
	Using Lagrange interpolation, we can recover the value of each $R'_i$ for random values of the variables $\{x_e\}_{e\in E(B_G)}$	(the values are the same for all the polynomials). 
	The interpolation algorithm requires $|I|\cdot D=\Oh(nMD)$ evaluations of $R'$.
	Since $R'=P'$, by Lemma~\ref{cor:poly-P-eval} each of them takes $2^n\poly(n+M))$ time.
	Our algorithm reports the minimum $w$ such that $R'_w$ evaluated to a non-zero element of $\field{2^t}$, or $+\infty$ if no such $w$ exists. 
	The solution of weight $w$ is then found using Lemma~\ref{lem:finding}.
	The event that the optimum value $w^*$ is not reported means that $R'_{w^*}$ is a non-zero polynomial that evaluated to 0 at the randomly chosen values.
	By Lemma~\ref{lem:zs} this happens with probability at most $\deg P / 2^t \le 1/2$, and one can make this probability arbitrarily small by standard methods. 
\end{proof}


Theorem~\ref{thm:bnd-tsp} follows now immediately by applying Theorem~\ref{thm:kernel} which reduces the general problem to the $M=\Oh(n^2)$ case and solving the resulting instance by Lemma~\ref{lem:bnd-tsp-small-demands}. Theorem~\ref{thm:bnd-tsp} says in particular that
if finite weights are bounded by a polynomial in $n$ then we can solve \probMVTSP in time $\Ohstar(2^n)$ and polynomial space by a randomized algorithm with no false positives and with false negatives with arbitrarily small constant probability. 

\section{The general case}
\label{sec:general}
In this section we prove Theorem~\ref{thm:expspace}, i.e., we show an algorithm solving \probMVTSP in time $\Ohstar(4^n)$.
In fact, we do not introduce a new algorithm, but we consider an algorithm by Berger et al.~(Algorithm 5 in~\cite{berger-soda}) and we provide a refined analysis, resulting in an improved running time bound $\Ohstar(4^n)$, which is tight up to a polynomial factor.

Let us recall the algorithm of Berger et al., in a slightly changed notation.
In fact, they solve a slightly more general problem, namely \probFixDegOutSub.
Let $I=(d,\In,\Out,r)$ be an instance of this problem, i.e., we want to find a solution $m:V^2\rightarrow\Zq$ that satisfies the degree constraints specified by $\In$ and $\Out$ and contains an outbranching rooted at $r$.
In what follows we assume $V=\{1,\ldots,n\}$ and $r=1$.

Consider an outbranching sequence $\{\delta_v\}_{v \in V}$ rooted at $r=1$.
In what follows, all outbranching sequences will be rooted at $1$, so we skip specifying the root.
Let $T_\delta$ be a minimum cost outbranching among all outbranchings with outdegree sequence $\delta$ and let $r_\delta$ be an optimum solution of \probFixDegSub for instance $(d,\In',\Out')$ where $\Out'=\Out-\outdeg_T$ and $\In'=\In-\indeg_T$.
Berger et al.\ note that then $m_\delta=m_{T_\delta}+r_\delta$ is a feasible solution for instance $I$ of \probFixDegOutSub, and moreover it has minimum cost among all solutions that contain an outbranching with outdegree sequence $\delta$.
Since $r_\delta$ can be found in polynomial time by Observation~\ref{obs:flows}, in order to solve instance $I$ it suffices to 
find outbranchings $T_\delta$ for all outbranching sequences $\delta$ and return the solution $m_\delta$ of minimum cost. Hence, Theorem~\ref{thm:expspace} boils down to proving the following lemma.

\begin{lemma}
	\label{lem:4^n}
	There is an algorithm which, for every outbranching sequence $\delta$, finds  a minimum cost outbranching among all outbranchings with outdegree sequence $\delta$ and runs in time $\Ohstar(4^n)$.
\end{lemma}

\newcommand{\procBestOut}{\ProblemName{BestOutbranching}}
\newcommand{\procMathBestOut}{\textsc{BestOutbranching}}
\newcommand{\first}{\textsf{first}}
\newcommand{\mincost}{\textsf{minCost}}
\newcommand{\best}{\textsf{best}}
\newcommand{\lastrem}{\textsf{lastRmvd}}
\newcommand{\bad}{\textsf{bad}}

We prove Lemma~\ref{lem:4^n} by using dynamic programming (DP). 
However, it will be convenient to present the DP as a recursive function \procBestOut with two parameters, $S\subseteq V$ and $\{\delta_v\}_{v\in S}$ (see Algorithm~\ref{alg:general}). 
It is assumed that $1\in S$.
We will show that $\procMathBestOut(S,\delta)$ returns a minimum cost out-tree among all out-trees with outdegree sequence $\delta$ that are rooted at $1$ and span $S$.
Our algorithm runs \procBestOut for $S=V$ and all outbranching sequences $\delta:V\rightarrow\Zq$.
Whenever \procBestOut returns a solution for an input $(S,\delta)$, it is memoized (say, in an efficient dictionary), so that when \procBestOut is called with parameters $(S,\delta)$ again, the output can be retrieved in polynomial time.

\begin{algorithm}
  \begin{algorithmic}
    
    \caption{} \label{alg:general}
    \Function {BestOutbranching} {$S, \delta$}
      \State $v_\first \gets \min\{v\in S\mid\delta_v = 0\}$
      \If {$|S| = 2$}
        \Return $\{(1,v_\first)\}$.
      \Else
        \State $\mincost \gets \infty$
        \For {$w \in S$}
          \If {$(\delta_w \ge 1 \wedge w \neq 1) \vee (\delta_w \ge 2 \wedge w = 1)$}
            \State $S' \gets S\setminus\{v_\first\}$
            \State $\delta' \gets \delta|_{S'}$
            \State $\delta'_w \gets \delta'_w - 1$ 
            \State $R_w \gets \procMathBestOut(S', \delta') \cup \{(w,v_\first)\}$
            \If {$d(R_w) < \mincost$}
              \State $\mincost \gets d(R_w)$
              \State $\best \gets R_w$
            \EndIf
          \EndIf
        \EndFor
        \Return $\best$ 
      \EndIf
    \EndFunction
  \end{algorithmic}
\end{algorithm}

Let us define $\lastrem(S) := \max (\{0, 1, 2, \ldots, n\} \setminus S)$
and $\bad(S, \delta) := \{v \in S \mid v < \lastrem(S) \wedge \delta_v = 0\}$.
Let us call $(S, \delta)$ a {\em reachable state} if it meets the following conditions:
\begin{enumerate}[$(i)$]
  \item $\delta_1 \ge 1$
  \item $\sum_{v\in S} \delta_v = |S| - 1$
  \item $|\bad(S,\delta)| \le 1$
\end{enumerate}

\begin{lemma}
  \label{lem:only-reachable}
  If function \procBestOut is given a reachable state as input then
  all recursively called \procBestOut will also be given only reachable states.
\end{lemma}

\begin{proof}
  Let us fix a reachable state $(S,\delta)$ for $|S|>2$ and consider the associated value $v_\first$ from the algorithm.
  Denote $S' = S\setminus\{v_\first\}$.
  Clearly, it suffices to show that all pairs $(S',\delta')$ created in the {\bf for} loop are reachable states.
  First, let us argue that $\bad(S', \delta)=\emptyset$.
  There are two cases:
  \begin{itemize}
    \item Assume $|\bad(S,\delta)| = 0$.
    In this case $v_\first > \lastrem(S)$ so $\lastrem(S') = v_\first$.
    Then,  $\bad(S', \delta)= \{v \in S' \mid v < \lastrem(S') \wedge \delta_v = 0\}=\{v \in S \mid v < v_\first \wedge \delta_v = 0\}=\emptyset$.  
    \item Assume $|\bad(S,\delta)| = 1$.
    Then, (1) $\lastrem(S') = \lastrem(S)$ because $\lastrem(S) > v_\first$ and (2) $\bad(S,\delta) = \{v_\first\}$.
    It follows that $\bad(S', \delta)\stackrel{(1)}{=}\{v \in S' \mid v < \lastrem(S) \wedge \delta_v = 0\} = \bad(S, \delta)\setminus\{v_\first\}\stackrel{(2)}{=}\emptyset$.
  \end{itemize}  
  Let us consider the recursive call of \procBestOut for a particular $w$.
  Sequence $\delta'|_{S'}$ differs from $\delta$ only at $w$, so
  $\bad(S',\delta') \subseteq \{w\} \cup \bad(S',\delta) = \{w\}$.
  This means that the condition $(iii)$ from the definition of reachable state holds for $(S',\delta')$.
  Since $(S,\delta)$ is reachable, $\delta_1\ge 1$.
  Then either $w\ne 1$ and $\delta'_1=\delta_1\ge 1$ or $w=1$ and $\delta'_1=\delta_1-1 \ge 1$, where the last inequality holds thanks to the condition in the {\bf if} statement in Algorithm~\ref{alg:general}.
  In both cases, $(i)$ holds for $(S',\delta')$.
  Finally, $(ii)$ is immediate by the definition of $\delta'$.
  It follows that $(S', \delta')$ is a reachable state, as required.
\end{proof}

\begin{lemma}
  \label{lem:correctness}
  If function \procBestOut is given a reachable state $(S,\delta)$, it returns a cheapest out-tree $T$ rooted at vertex $1$, spanning $S$ and with outdegree sequence $\delta$.
\end{lemma}
\begin{proof}
  We will use induction on $|S|$.
  
  In the base case $|S| = 2$, there is only one outbranching spanning $S$ rooted at $1$, namely $\{(1, v_\first)\}$ and it is indeed returned by the algorithm.

  In the inductive step assume $|S| > 2$.
  By conditions $(i)$ and $(ii)$ in the definition of a reachable state and Lemma~\ref{lem:out-sequence}, there is at least one out-tree rooted at $1$, spanning $S$, and with outdegree sequence $\delta$.
  Let $T$ be a cheapest among all such out-trees. 
  Vertex $v_\first$ is a leaf of $T$, since $\delta_{v_\first}=0$.
  At some point $w$ in the {\bf for} loop in Algorithm~\ref{alg:general} is equal to the parent $w^*$ of $v_\first$ in $T$.
  Then, $T \setminus \{(w^*, v_\first)\}$ is an out-tree rooted at $1$, spanning $S'$, and with outdegree sequence $\delta'$.
  Since $(S',\delta')$ is a reachable state by Lemma~\ref{lem:only-reachable}, by the inductive hypothesis we know that a cheapest such out-tree $T'$ will be returned by $\procMathBestOut(S', \delta')$. 
  In particular, it means that $d(T')\le d(T\setminus\{(w^*,v_\first)\})$.
  Denote $R_{w^*}=T' \cup \{(w^*,v_\first)\}$.
  Then, $d(R_{w^*})=d(T')+d(w^*,v_\first)\le d(T\setminus\{(w^*,v_\first)\})+d(w^*,v_\first)=d(T)$.
  It follows that \procBestOut returns a set of edges $\best$ of cost at most $d(T)$.
  However $\best=R_w$ for a vertex $w$ and by applying the induction hypothesis it is easy to see that $R_w$ is an out-tree rooted at $1$, spanning $S$ with outdegree sequence $\delta$. The claim follows.
\end{proof}

\newcommand{\ddelta}{\bar{\delta}}
\newcommand{\bs}{\bar{s}}

\begin{lemma}
  \label{lem:few-reachable}
  There are $\Ohstar(4^n)$ reachable states.
\end{lemma}
\begin{proof}
  Any sequence of $n$ nonnegative integers that sums up to at most $n-1$ will be called an {\em extended sequence}.
  It is well known that there are exactly ${{2n - 1} \choose {n}} < 2^{2n - 1} = \Oh(4^n)$ such sequences.
To see this consider sequences of $n-1$ balls and $n$ barriers and bijectively map them to the sequences of $n$ numbers by counting balls between barriers and discarding the balls after the last barrier.
	
  Let us fix an extended sequence $\ddelta=\{\ddelta_v\}_{v \in V}$, and
  denote $\bs := n - (1 + \sum_{i = 1}^n \ddelta_i)$.
  We claim that there are only $\Oh(n)$ reachable states $(S,\delta)$ such that $\ddelta|_S=\delta$ and $\ddelta|_{V\setminus S}=0$.
  Consider any such pair $(S,\delta)$.
  Let $(v_1, v_2, \ldots, v_k)$ be the vertices of $\{v\in V\mid \ddelta_v=0\}$ sorted in increasing order.
  By the definition of a reachable state we know that $|S|=1+\sum_{i = 1}^n \ddelta_i$, so $\bs = |\{1, 2, \ldots, n\} \setminus S|$.
  By $(ii)$, for at least one vertex $v\in S$ we have $\ddelta_v=\delta_v=0$, so $k\ge\bs+1$.
  Let us assume that $k\ge \bs+2$ and $\lastrem(S) \ge v_{\bs+2}$. 
  Then, $\{v_1, v_2, \ldots, v_{\bs+1}\} \cap S \subseteq \bad(S,\delta)$.
  Since $v_{\bs+2} \le \lastrem(S) \not\in S$, at most $\bs-1$ elements from $\{v_1, v_2, \ldots, v_{\bs+1}\}$ are outside $S$, so $|\bad(S,\delta)| \ge (\bs+1) - (\bs-1) = 2$.
  This is a contradiction with $(S,\delta)$ being a reachable state, which proves that $k\le \bs+1$ or $\lastrem(S) < v_{\bs+2}$.
  In any case, $\{1, 2, \ldots, n\} \setminus S \subseteq \{v_1, \ldots, v_{\bs+1}\}$.
  There are $\bs+1 = \Oh(n)$ ways to choose $\bs$ elements to the set $\{1, 2, \ldots, n\} \setminus S$ from $\{v_1, \ldots, v_{\bs+1}\}$, so equivalently there are $\Oh(n)$ sets $S$ such that $(S,\delta)$ is a reachable state, $\ddelta|_S=\delta$ and $\ddelta|_{V\setminus S}=0$.
  
  Every reachable state $(S,\delta)$ has the corresponding extended sequence $\{\ddelta\}_{v\in V}$ defined by $\ddelta|_S=\delta$ and $\ddelta|_{V\setminus S}=0$.
  Since there are $\Oh(4^n)$ extended sequences, and each of them has $O(n)$ corresponding reachable states there are $\Oh(4^n) \cdot \Oh(n) = \Ohstar(4^n)$ reachable states in total.
\end{proof}

We are ready to prove Lemma~\ref{lem:4^n}.
Recall that our algorithm runs $\procMathBestOut(V,\delta)$ for all outbranching sequences $\delta$ and uses memoization to avoid repeated computation.
We claim that for any outbranching sequence $\delta$, the pair $(V,\delta)$ is a reachable state .
Indeed, conditions $(i)$ and $(ii)$ hold since $\delta$ is an outbranching sequence.
By definition, $\lastrem(V)=0$, so $\bad(V,\delta)=\emptyset$ which implies $(iii)$.
Hence by Lemma~\ref{lem:correctness} the algorithm is correct.
By Lemma~\ref{lem:only-reachable} the running time can be bounded by the number of reachable states times a polynomial, which is $\Ohstar(4^n)$ by Lemma~\ref{lem:few-reachable}.
This ends the proof of Lemma~\ref{lem:4^n} and hence also Theorem~\ref{thm:expspace}, as discussed in the beginning of this section.

\section{Polynomial space}
\label{sec:polyspace}
\newcommand{\cost}{{\sf cost}}
\newcommand{\capa}{{\sf cap}}

In this section we show Theorem~\ref{thm:polyspace}, that is, we solve \probMVTSP in $\Ohstar(7.88^n)$ time and polynomial space. Berger et al.~\cite{berger-arxiv} solved this problem in $\Oh(16^{n+o(n)})$ time and polynomial space, with the key ingredient being the following.

\begin{lemma}[Berger et al.~\cite{berger-arxiv}]
	\label{lem:berger:outbranching}
	There is a polynomial space algorithm running in time $\Oh(4^{n+o(n)})$ which, given an outdegree sequence $\{\delta_v\}_{v\in V}$, a cost function $d:V^2\rightarrow\Zq$, and a root $r\in V$ computes the cheapest outbranching rooted at $r$ with the required outdegrees.
\end{lemma}

More precisely, the $\Oh(16^{n+o(n)})$-time algorithm consists of the following steps:
\begin{enumerate}[$(i)$]
	\item Enumerate all $\Oh(4^n)$ outbranching sequences
	\item For each outbranching sequence compute the cheapest outbranching with required degrees using Lemma~\ref{lem:berger:outbranching} in time $\Oh(4^{n+o(n)})$
	\item For each of these outbranchings complete it to a solution of the original \probMVTSP instance with an optimal solution of \probFixDegSub on the residual degree sequences (in polynomial time, by Observation~\ref{obs:flows}).
\end{enumerate}

The intuition behind our approach is as follows.
We iterate over all subsets of vertices $R$.
Here, $R$ represents our guess of the set of inner vertices of an outbranching in an optimal solution.
Then we perform $(i)$ and $(ii)$ in the smaller subgraph induced by $R$.
Finally, we replace $(iii)$ by a more powerful flow-based algorithm which connects the vertices in $V\setminus R$ to $R$, and at the same time computes a feasible solution of \probFixDegSub on the residual degree sequences, so that the total cost is minimized.
Let $r=|R|$.
Clearly, when $r$ is a small fraction of $n$, we get significant savings in the running time.
The closer $r/n$ is to $1$ the smaller are the savings, but also the smaller is the number ${n \choose r}$ of sets $R$ to examine.

In fact, the real algorithm is slightly more complicated.
Namely, we fix an integer parameter $K$ and $R$ corresponds to the set of vertices left from an outbranching in an optimal solution after $K$ iterations of removing all leaves.
The running time of our algorithm depends on $K$, because the algorithm actually guesses the layers of leaves in each iteration. The space complexity is polynomial and does not depend on $K$.
In the end of this section, we show that our running time bound is minimized when $K=4$.

\subsection{Our algorithm}

Similarly as in Section~\ref{sec:general}, we solve the more general \probFixDegOutSub: for a given instance $I=(d,\In,\Out,\myroot)$ we want to find a solution $m:V^2\rightarrow\Zq$ that satisfies the degree constraints specified by $\In$ and $\Out$ and contains an outbranching rooted at $\myroot$.

Let $T$ be an arbitrary outbranching.
We define a sequence $L_1(T), L_2(T), \ldots$ of subsets of $V(T)$ as follows. 
For $i \ge 1$ let $L_i(T)$ be the set of leaves of $T \setminus (L_1(T) \cup L_2(T) \cup \ldots \cup L_{i-1}(T))$ if $|V(T) \setminus (L_1(T) \cup \ldots \cup L_{i-1}(T))|>1$, and otherwise $L_i=\emptyset$.
The sets $L_i(T)$ will be called \emph{leaf layers}.
Denote $R_i(T)=V\setminus(L_1(T)\cup\cdots\cup L_i(T))$ for any $i\ge 1$.

\begin{lemma}
	\label{obs:leaf-layers}
	For every $i\ge 1$ we have $\myroot \in R_i(T)\setminus L_{i+1}(T)$, $|L_i(T)| \ge |L_{i+1}(T)|$ and $|L_{i+1}|\le\frac{n-|R_i(T)|}{i}$.
\end{lemma}

\begin{proof}
	In this proof we skip the `$(T)$' in $L_i$ and $R_i$ because there is no ambiguity.
	Assume $\myroot\in L_i$ for some $i\ge 1$.
	It means that $\myroot$ is a leaf in $T \setminus (L_1 \cup L_2 \cup \ldots \cup L_{i-1})$.
	Then $V \setminus (L_1 \cup L_2 \cup \ldots \cup L_{i-1})=\{\myroot\}$ and $L_i=\emptyset$, a contradiction. 
	Hence $\myroot \not\in L_i$ for all $i\ge 1$, and in consequence  $\myroot \in R_i$ for all $i\ge 1$.
	However, $\myroot \in R_{i+1}(T)$ implies that $\myroot \not\in L_{i+1}(T)$, hence $\myroot \in R_i\setminus L_{i+1}$.
	
	If $|V \setminus (L_1 \cup \ldots \cup L_{i})|>1$, then $L_{i+1}$ is the set of leaves of the out-tree $T \setminus (L_1 \cup \ldots \cup L_{i})$, which is contained in the set of parents of vertices in $L_i$. Since every vertex in $L_i$ has exactly one parent, $|L_i|\ge |L_{i+1}|$.
	If $|V \setminus (L_1 \cup \ldots \cup L_{i})|\le 1$ then $L_{i+1}=\emptyset$ and clearly $|L_i|\ge |L_{i+1}|=0$.
	
	Finally, since for every $j<i$ we have $|L_j|\ge |L_i|$ we get $n-|R_i|=|L_1|+\ldots+|L_i|\ge i|L_i|$. It follows that $|L_{i+1}|\le |L_i| \le \frac{n-|R_i|}i$, as required.
\end{proof}

A pseudocode of our algorithm is presented as Algorithm~\ref{alg:polyspace}.

\begin{algorithm}
	\begin{algorithmic}[1]
		
		\caption{} \label{alg:polyspace}
		\Function {Solve} {$G, \Out, \In, d, \myroot$}
		\State $\best \gets \infty$
		\For {$R, L_{K+1}, \delta$}
		\State $T_R \gets$ cheapest $\delta$-out-tree spanning $R$ rooted at $\myroot$ (Lemma~\ref{lem:berger:outbranching})
		\State $\Out' \gets \Out-\outdeg_{T_R}$
		\State $\In' \gets \In-\indeg_{T_R}$
		\For {$L_1, \ldots, L_{K}$}
		\State $F \gets \textsc{CreateNetwork}(G, R, \Out', \In', d, L_1, \ldots, L_K)$ 
		\State $f \gets \textsc{MinCostMaxFlow}(F)$ 
		\If {$|f| = \sum_{v \in V(G)} \Out'(v)$ {\bf and} $\cost(f) + d(T_R) < \best$ }
		\State $\best \gets \cost(f) + d(T_R)$
		\EndIf
		\EndFor
		\EndFor
		\Return $\best$
		\EndFunction
		
		%

	\end{algorithmic}
\end{algorithm}

For clarity, in the pseudocode we skipped some constraints that we enforce on the sets $L_i$ and sequence $\delta$. We state them below.
\begin{enumerate}[(C1)]
	\item $L_{K+1} \subseteq R \subseteq V, \myroot \in R \setminus L_{K+1}, |L_{K+1}| \le \frac{n-|R|}{K}$
	\item $\{\delta_v\}_{v\in R}$ is a rooted out-tree sequence, i.e., for all $v\in R$ we have $\delta_v \in \Zq, \sum_{v \in R} \delta(v) = |R| - 1$; 
	also $\delta_\myroot\ge 1$ if $|R|\ge 2$ and $\delta_\myroot=0$ if $|R|=1$.
	\item for every $v \in L_{K+1}$ we have $\delta_v = 0$ and for every $v \in R\setminus (L_{K+1}\cup\{\myroot\})$ we have $\delta_v \ge 1$
	\item $L_1\uplus L_2\uplus\cdots\uplus L_K=V\setminus R$
	\item $|L_i| \ge |L_{i+1}|$ for $i=1,\ldots,K$.
\end{enumerate}
It is clear that all these possibilities can be enumerated in time proportional to their total number times $O(n)$.

Let us provide some further intuition about Algorithm~\ref{alg:polyspace}.
Consider an optimum solution $m$ of $I$ and any outbranching $B$ in $m$ rooted at $\myroot$.
In Algorithm~\ref{alg:polyspace}, for any $i=1,\ldots K+1$, the set $L_i$ is a guess of the leaf layer $L_i(B)$, while $R$ is a guess of $V\setminus(L_1(B)\cup\cdots\cup L_K(B))$. Finally, $\delta$ is a guess of the outdegree sequence of the out-tree $B[R]$. 

In Line 8 we create a flow network, and in line 9 a minimum cost maximum flow is found in polynomial time. In the next section we discuss the flow network and properties of the flow.

\subsection{The flow}

In this section we consider a run of Algorithm~\ref{alg:polyspace}, and in particular we assume that the variables $R,\delta,L_1,\ldots,L_{K+1}$ have been assigned accordingly.
Function {\sc CreateNetwork} in our algorithm builds a flow network $F = (V(F), E(F), \capa, \cost)$, where $E(F)$ is a set of directed edges and $\capa$ and $\cost$ are functions from edges to integers denoting capacities and costs of corresponding edges. 
As usual, the function $\cost$ extends to flow functions in a natural way, i.e., $\cost(f)=\sum_{e\in E(F)}f(e)\cost(e)$.
We let $V(F) = \{s, t\} \cup \{v^I, v^O \mid v \in V(G)\} \cup \{v^C \mid v \in V\setminus R\}$, where $s$ and $t$ denote the source and the sink of $F$.

We put following edges into $E(F)$:
\begin{itemize}
	\item $(s, v^O)$, where $\capa(s, v^O) = \Out'(v), \cost(s, v^O) = 0$ for every $v \in V(G)$
	\item $(v^I, t)$, where $\capa(v^I, t) = \In'(v), \cost(v^I, t) = 0$ for every $v \in R$ 
	\item $(v^I, t)$, where $\capa(v^I, t) = \In'(v)-1, \cost(v^I, t) = 0$ for every $v \notin R$
	\item $(v^C, t)$, where $\capa(v^C, t) = 1, \cost(v^C, t) = 0$ for every $v \notin R$
	\item $(u^O, v^I)$, where $\capa(u^O, v^I) = \infty, \cost(u^O, v^I) = d(u, v)$ for every $(u, v) \in E(G)$
	\item $(u^O, v^C)$, where $\capa(u^O, v^C) = \infty, \cost(u^O, v^C) = d(u, v)$ for every $v \in L_i, u \in R \cup L_{i + 1} \cup \ldots \cup L_K, (u, v) \in E(G)$. 
\end{itemize}

We will say that $F$ has a \emph{full flow} if it has a flow $f$ with value $|f|=\sum_{v \in V} \Out'(v)$. By the construction of $F$, then all edges leaving source are saturated, i.e., carry flow equal to their capacity. Since $\sum_{v \in V} \Out'(v)=\sum_{v \in V} \In'(v)$, also all edges that enter the sink are saturated.

Essentially, the network above results from extending the standard network used to get Observation~\ref{obs:flows} by vertices $v^C$. 
The flow between $\{v^O\mid v\in V\}$ and $\{v^I\mid v\in V\}\cup \{v^C \mid v \in V\setminus R\}$ represents the resulting solution. 
In a full flow the edges leaving $v^C$ are saturated, so a unit of flow enters every vertex $v^C$, which results in connecting $v$ in the solution to a higher layer or to $R$.
Thanks to that the solution resulting from adding the out-tree $T_R$ to the solution extracted from $f$ contains an outbranching.

\begin{lemma} \label{clm:flow-to-sol}
	If $f$ is a full flow of minimum cost in $F$ then there exists a solution of $I$ with cost $\cost(f) + d(T_R)$.
	Moreover, the solution can be extracted from $f$ in polynomial time.
\end{lemma}

\begin{proof}
	By standard arguments, since all capacities in $F$ are integer, we infer that there is an integral flow of minimum cost (and it can be found in polynomial time), so we assume w.l.o.g.\ that $f$ is integral.
	
	Let $b : V^2 \to \{0, 1\}$ denote a function such that $b(u, v) = [(u, v)\in T_R]$.
	Now we construct a solution $m:V^2\rightarrow\Zq$ of $I$.
	
	 \[
	m(u,v) =\begin{cases}
	f(u^O, v^I) + b(u, v) & \text{if } v \in R\\
	f(u^O, v^I) + f(u^O, v^C) & \text{if } v \not\in R.
	\end{cases}
	\]
	
	In other words, $m$ describes how many times edge $(u, v)$ was used by the out-tree $T_R$ and flow $f$ in total. Let us verify that $m$ is a feasible solution for $I$.
	The degree constraints are easy to verify, so we are left with showing that $m$ contains an outbranching rooted at $\myroot$.
	To this end it suffices to show that every vertex $v$ is reachable from $\myroot$ in $G_m$.
	Clearly, this holds for vertices in $R$, thanks to the out-tree $T_R$.
	Pick an arbitrary vertex $v\not\in R$.
	Then $v \in L_i$ for some $i=1,\ldots,K$.
	We know that $f(v^C, t) = 1$, so there exists $u$ such that $f(u^O, v^C) = 1$. 
	Therefore, $v$ is connected in $G_m$ to a vertex from $R \cup L_{i+1} \cup \ldots L_{K}$. Since $v$ in $G_m$ has an in-neighbor either in $R$ or in a layer with a higher index,
	we can conclude that there is a path in $G_m$ from $R$ to $v$. 
	Hence $m$ indeed contains the required outbranching.
	
	Finally, it can be easily checked that $d(m)=\cost(f) + d(T_R)$, what concludes this proof.
\end{proof}

Let $m$ be a feasible solution for $I$.
Let $R$, $L_i$ for $i=1,\ldots,K+1$ be sets of vertices and $\delta$ an out-tree sequence on $R$, 
as in Algorithm~\ref{alg:polyspace}.
We say that $m$ is {\em compliant} with $R$, $L_1,\ldots,L_{K+1}$ and $\delta$ when $m$ contains an outbranching $T$ rooted at $\myroot$ such that $R_K(T)=R$, $L_i(T)=L_i$ for $i=1,\ldots,K+1$ and $\delta$ is equal to the outdegree sequence of $T[R]$.

\begin{lemma} \label{clm:sol-to-flow}
	Assume that there exists a solution $m$ of $I$ that is compliant with $R$, $L_1,\ldots,L_{K+1}$ and $\delta$. Then $F$ has a full flow $f$ such that $\cost(f) + d(T_R) \le d(m)$.
\end{lemma}

\begin{proof}
	Let $T$ be an outbranching in $m$ which certifies that $m$ is compliant with $R$, $L_1,\ldots,L_{K+1}$ and $\delta$.
	Let $p : V^2 \to \{0, 1\}$ be a function such that for every $u,v\in V$ we have $p(u,v)=[(u,v) \in T]$.
	
	We set $f(s, u) = \capa(s, u)$ for all edges $(s,u) \in E(F)$ and 
	       $f(u, t) = \capa(u, t)$ for all edges $(u,t) \in E(F)$.
	If $v \in V\setminus R$ then we set $f(u^O, v^C) = p(u, v)$. 
	For all $u, v \in V(G)$ we set $f(u^O, v^I) = m(u, v) - p(u, v)$. 
	It can be easily checked that such function $f$ is a full flow and $\cost(f)=d(m) - d(T[R])$.
	However, since $T[R]$ is a $\delta$-out-tree rooted at $\myroot$ and $T_R$ is a cheapest such out-tree, $d(T_R)\le d(T[R])$.
	It follows that $\cost(f)\le d(m) - d(T_R)$, so $\cost(f) + d(T_R) \le d(m)$ as required.
\end{proof}

Consider a  {\em minimum cost} full flow $f'$ in $F$ that is found by Algorithm~\ref{alg:polyspace} for a choice of $R, L_1,\ldots,L_{K+1}, \delta$. 
The claim above implies that $\cost(f') + d(T_R) \le d(m)$.
However, notice that we do not claim that $\cost(f')$ is the cost of optimal completion of $T_R$ consistent with all guesses, as the intuitions we described earlier might suggest. 
It could be the case that in the solution resulting from $f'$, a vertex which was guessed to belong to $L_i$ does not have any out-neighbor that was guessed to belong to $L_{i-1}$, what would mean that this vertex should be in an earlier layer.
However, that is not an issue for the extraction of the global optimum solution of $I$, because we may get only better solutions than the optimum completion for that particular guess. 

\subsection{Correctness}

\begin{lemma}
	Function {\sc Solve} returns the cost of an optimal solution of $I$.
\end{lemma}

\begin{proof}
	From Lemma~\ref{clm:flow-to-sol} we infer that {\sc Solve} returns the cost of a feasible solution of $I$. 
	It remains to show that it returns a value that smaller or equal to the cost of an optimal solution of $I$. 
	To this end, let $m$ be an arbitrary optimal solution of $I$ and let $T$ be an arbitrary outbranching rooted at $\myroot$ in $G_m$.
	Let $R=R_K(T)$, $L_i=L_i(T)$ for $i=1,\ldots,K+1$ and let $\delta$ be the outdegree sequence of $T[R]$.
	
	Let us verify that $R, L_1,\ldots,L_{K+1}$ and $\delta$ satisfy constraints (C1)--(C5).
	We get (C1) and (C5) by Lemma~\ref{obs:leaf-layers}.
	(C2) follows from the definition of $\delta$.
	For (C3), consider two cases.
	If $|R|>1$, then  $L_{K+1}$ is the set of leaves in $R$ and hence indeed for every $v \in L_{K+1}$ we have $\delta_v = 0$ and for every $v \in R\setminus (L_{K+1}\cup\{\myroot\})$ we have $\delta_v \ge 1$.
	When $|R|\le 1$, we have $L_{K+1}=\emptyset$ and since $\myroot\in R$ by Lemma~\ref{obs:leaf-layers}, $R=\{\myroot\}$. Then both sets $L_{K+1}$ and $R\setminus (L_{K+1}\cup\{\myroot\})$ are empty, so (C3) trivially holds.
	Finally, (C4) follows by the definition of leaf layers.

	Since $R, L_1,\ldots,L_{K+1}$ and $\delta$ satisfy constraints (C1)--(C5), then {\sc Solve} reaches this particular evaluation of the variables $R, L_1,\ldots,L_{K+1}$ and $\delta$.
	Then, based on Lemma \ref{clm:sol-to-flow}, the network $F$ has a full flow $f$ such that $\cost(f) + d(T_R) \le d(m)$, and it follows that {\sc Solve} returns a value $\best\le \cost(f) + d(T_R) \le d(m)$, as required.
\end{proof}

Obviously, {\sc Solve} can be easily adapted to return a solution of $I$ with the cost it returns, but we have not taken this into account in Algorithm~\ref{alg:polyspace} for the sake of its readability.

\subsection{Running time}

Having a correct algorithm solving \probFixDegOutSub in polynomial space, let us analyze its complexity depending on $K$.

Let us denote $r=|R|$ and $c=|L_{K+1}|$.
Recall that $1\le r \le n$ and $0\le c\le \lfloor \frac{n - r}{K} \rfloor$.

If we fix $r$ and $c$, then there are ${n-1 \choose r-1}$ guesses for $R$ (it has to contain $\myroot$) and at most ${r-1 \choose c}$ guesses for $L_{K+1}$.
Let us bound the number of guesses for $\delta$.
By (C2) and (C3), $\sum_{v\in R}\delta_v=r-1$, and $\delta_v=0$ iff $v\in L_{K+1}$ so essentially we put $r-1$ balls into $r-c$ bins that must be nonempty, which is ${r - 2 \choose c - 1}$ by standard combinatorics. In the special case $c=0$ there is one choice for $\delta$, where $\delta_{\myroot} = 0$.

In total, there are at most ${n \choose r} {r \choose c}^2$ guesses for all $R, L_{K+1}, \delta$ simultaneously. For each of these guesses, using Lemma~\ref{lem:berger:outbranching} function {\sc Solve} calculates an optimal $\delta$-out-tree spanning $R$, which takes time $\Oh(4^{n+o(n)})$. 
It follows that that part takes time 
$\Ohstar(\sum_r {n \choose r} 4^{r+o(r)} \sum_c {r \choose c}^2).$
Then, {\sc Solve} guesses a partition of $V \setminus R$ into $L_1, \ldots, L_K$ in at most $K^{n-r}$ ways. 
For each such guess, {\sc Solve} spends polynomial time, so that part takes $\Ohstar (\sum_r {n \choose r} K^{n-r} \sum_c  {r \choose c}^2)$ time.
Hence the total running time can be bounded by $$\Ohstar  \left( 2^{o(n)}\sum_{r=1}^{n} \sum_{c=0}^{\lfloor \frac{n - r}{K} \rfloor}  \underbrace{{n \choose r} (K^{n-r} + 4^r) {r \choose c}^2}_{\xi(r,c)} \right) .$$

Since there are polynomially many guesses for $r$ and $c$, we can actually replace sums with maximums in the expression above and focus on the expression $\xi(r,c)={n \choose r} (K^{n-r} + 4^r) {r \choose c}^2$.

We will heavily use the well-known bound ${n \choose \alpha n} < 2^{h(\alpha)n}$, where $h(\alpha)=-\alpha\log_2\alpha-(1-\alpha)\log_2(1-\alpha)$ is the binary entropy function (see e.g.~\cite{expalg-book}). 
For readability, let us denote $f(\alpha) = 2^{h(\alpha)}$ and let us point out that $f$ is increasing on interval $[0, \frac{1}{2}]$ and decreasing on interval $[\frac12, 1]$. 
Let us denote $\beta \coloneqq \frac{r}{n}$.
We are going to distinguish two cases here.

\begin{enumerate}
	\item $\frac{n-r}{K} \ge \frac{r}{2}$
	
	This inequality can be rephrased as $r \le \frac{2}{K+2} n$, which is equivalent to $\beta \le \frac{2}{K+2}$. 
	We will use here a trivial bound ${r \choose c} \le 2^r$.
	Then, $\xi(r,c)\le f(\beta)^n ((K^{1-\beta})^n + 4^{\beta n}) 4^{\beta n} =
	(f(\beta)K^{1-\beta}4^{\beta})^n + (f(\beta)4^{2\beta})^n$
	
	\item $\frac{n-r}{K} < \frac{r}{2}$
	
	In that case we know that $\max_{c=0}^{\lfloor \frac{n-r}{K} \rfloor} {r \choose c}^2$
	is attained when $c = \lfloor \frac{n-r}{K} \rfloor$ and for that particular value of $c$ we can use the following bound.
	$${r \choose c}^2 = {r \choose \frac{\lfloor \frac{n-r}{K} \rfloor}{r} \cdot r}^2 =
	\Ohstar \left( f\left( \frac{\lfloor \frac{n-r}{K} \rfloor}{r} \right) ^{2r}\right) =
	\Ohstar \left( f\left( \frac{\frac{n-r}{K}}{r} \right) ^{2r}\right) =$$
	$$	=\Ohstar \left( f \left( \frac{1 - \beta}{K \beta} \right) ^{2\beta n} \right) $$
	In the third equality above we used fact that $f$ is increasing in interval $[0, \frac12]$. 
	To sum up, in this case, 
	$$\xi(r,c)=\Ohstar \left( \left( f(\beta)K^{1-\beta}  f \left( \frac{1 - \beta}{K \beta} \right) ^{2\beta } \right)^n +
	\left( f(\beta)4^{\beta}  f \left( \frac{1 - \beta}{K \beta} \right) ^{2\beta } \right)^n \right).$$

\end{enumerate}

Our numerical analysis shows that it is optimal to choose $K=4$. 
For that particular value of $K$, the first case applies if and only if $\beta \le \frac13$. 
Then,
\[\xi(r,c)=(4f(\beta))^n+(4^{2\beta}f(\beta))^n=\Oh((4f(\beta))^n)=\Oh((4f(\tfrac13))^n)=\Oh(7.56^n).\]
Let us now investigate the second case, when $\beta > \frac13$.
For $\frac13 < \beta < \frac12$ the summand $ \left( f(\beta)4^{1-\beta}  f \left( \frac{1 - \beta}{4 \beta} \right) ^{2\beta } \right)^n$ dominates, and for $\beta \ge \frac12$ the summand 
$\left( f(\beta)4^{\beta}  f \left( \frac{1 - \beta}{4 \beta} \right) ^{2\beta } \right)^n$ dominates. 
We have numerically verified that 
\[f(\beta)4^{1-\beta}  f \left( \frac{1 - \beta}{4 \beta} \right) ^{2\beta } \le 7.68 \text{ for } \beta \in (\tfrac13,\tfrac12)\]
 and 
 \[f(\beta)4^{\beta}  f \left( \frac{1 - \beta}{4 \beta} \right) ^{2\beta } \le 7.871 \text{ for } \beta\in[\tfrac12,1].\]
Hence, we can conclude that for $K=4$ and our algorithm runs in time $\Ohstar(7.871^{n+o(n)})=\Oh(7.88^n)$ and in polynomial space. 
This concludes the proof of Theorem~\ref{thm:polyspace}.


\section{$(1+\epsilon)$-approximation}
\label{sec:approx}

In this section we show theorem \ref{thm:aprox}, i.e. we present an algorithm for \probMVTSP which finds a $(1+\epsilon)$-approximation in $\Ohstar\left(\frac{2 ^ n}{\epsilon}\right)$ time and polynomial space.

To achieve this we consider a more general problem, namely \probFixDegConSub.
The main idea is to round weights of edges of the given instance, so that we can use the algorithm for polynomially bounded weights from Lemma \ref{lem:bnd-tsp-small-demands} which is an analog of theorem \ref{thm:bnd-tsp} for \probFixDegConSub.

Let us first consider the case with degrees bounded by a polynomial.

\begin{lemma}
	\label{lem:aprx_bounded_degree}
	For given $\epsilon > 0$ and an instance $I=(d,\In,\Out)$ of \probFixDegConSub such that $\In(v), \Out(v) \le \Oh(n^2)$ for every vertex $v$
	there exists an algorithm finding a $(1+\epsilon)$-approximate solution in $\Ohstar\left(\frac{2 ^ n}{\epsilon}\right)$ time and polynomial space.
\end{lemma}

\begin{proof}
Let us denote the optimal solution for $I$ by $\OPT$.
First, our algorithm guesses the most expensive edge used by $\OPT$.
Let us denote its cost by $E$, in particular 
\setcounter{equation}{0}
\begin{equation}
	\label{eq:expensEdge}
 	E \le d(\OPT).
\end{equation}

Let us denote by $C$ the universal constant such that $\In(v), \Out(v) \le C n^2$ for every vertex $v$
and let us round $d$ in the following way
\begin{equation}
	\label{eq:def}
	d'(u,v) := \left\{ \begin{array}{ll}
	\ceil{\frac{C n^3}{\epsilon E}d(u,v)} & \textrm{if $d(u,v)\le E$}\\
	\infty & \textrm{if $d(u,v) > E$}
	\end{array} \right.
\end{equation}

If $d'(u,v)$ is finite then it is bounded by $\ceil{\frac{C n^3}{\epsilon E} E} = \ceil{\frac{C n^3}{\epsilon}}$. Our algorithm simply returns the optimal solution for instance $I' = (d',\In,\Out)$ which can be found in $\Ohstar\left(\frac{2 ^ n}{\epsilon}\right)$ time using the algorithm from Lemma \ref{lem:bnd-tsp-small-demands} with $D = \ceil{\frac{C n^3}{\epsilon}}$.
Let us denote this solution by $\ALG$.
Now we only need to prove that $\ALG$ is $(1 + \epsilon)$-approximation for the original instance $I$.
We know that $\ALG$ is an optimal solution for $I'$, in particular 
\begin{equation}
	\label{eq:optInD'}
	d'(\ALG)\le d'(\OPT).
\end{equation}

For every $v$ we have $\Out(v) \le C n^2$, so 
\begin{equation}
	\label{eq:cntEdges}
	\sum_{(u,v) \in V^2} \OPT(u,v) = \sum_{u \in V} \Out(u) \le n \cdot C n^2.
\end{equation}

The following chain of inequalities finishes the proof.
$$d(\ALG) \stackrel{(\ref{eq:def})}{\le} \frac{\epsilon E}{C n^3} d'(\ALG) \stackrel{(\ref{eq:optInD'})}{\le} \frac{\epsilon E}{C n^3} d'(\OPT) \stackrel{(\ref{eq:def})}{\le} \frac{\epsilon E}{C n^3} \sum_{(u,v) \in V^2} \OPT(u,v) \left(\frac{d(u,v) C n^3}{\epsilon E} + 1\right) = $$
$$ = d(\OPT) + \frac{\epsilon E}{C n^3} \sum_{(u,v) \in V^2} \OPT(u,v) \stackrel{(\ref{eq:cntEdges})}{\le} d(\OPT) + \frac{\epsilon E}{C n^3} Cn^3 \stackrel{(\ref{eq:expensEdge})}{\le} (1 + \epsilon) d(\OPT) $$

\end{proof}

Now we can generalize the algorithm from Lemma \ref{lem:aprx_bounded_degree} using it as a black box for the general case.

\begin{lemma}
	\label{lem:aproxConSub}
	For a given $\epsilon > 0$ and an instance $I=(d,\In,\Out)$ of \probFixDegConSub
	there exists an algorithm finding a $(1+\epsilon)$-approximate solution in $\Ohstar\left(\frac{2 ^ n}{\epsilon}\right)$ time and polynomial space.
\end{lemma}
\begin{proof}
	First let us use algorithm from Theorem \ref{thm:kernel} which outputs an instance $I'=(d,\In',\Out')$ of \probFixDegConSub and a function $f:V^2\rightarrow\Zq$.
	Let us denote the optimal solution for $I'$ by $\OPT'$. By Theorem \ref{thm:kernel} the optimal solution for $I$ equals $\OPT' + f$. Moreover, we know that $\In', \Out'(v) \le \Oh(n^2)$ for every vertex $v$.
	In particular we can use algorithm from Lemma \ref{lem:aprx_bounded_degree} to get a solution $\ALG'$ for instance $I'$ such that $d(\ALG') \le (1 + \epsilon) d(\OPT')$.
	Our algorithm simply returns solution $\ALG' + f$, which is a solution for $I$ because $\ALG$ is connected and $f$ increases degrees exactly by difference between $I$ and $I'$.
	To prove $\ALG' + f$ is $(1+\epsilon)$-approximation we just need to observe that
	$$d(\ALG' + f) \le (1 + \epsilon)d(\OPT') + d(f) \le (1 + \epsilon)d(\OPT' + f).$$
\end{proof}

\probFixDegConSub is a generalization of \probMVTSP so the algorithm from Lemma \ref{lem:aproxConSub} proves Theorem \ref{thm:aprox}.

\section{Further Research}
\label{sec:further}

Since TSP is solvable in time $\Ohstar(2^n)$ and exponential space~\cite{BellmanTSP,HeldKarpTSP} and time $\Oh(4^{n+o(n)})$ and polynomial space~\cite{GurevichShelah}, the main remaining question is whether these bounds can be achieved for \probMVTSP avoiding in the running time bound the linear dependence on maximum distance $D$.
Another interesting goal is a deterministic version of Theorem~\ref{thm:bnd-decision}.

\anonymyze{
\section*{Acknowledgements}
The research leading to the results presented in this paper was partially carried out during the Parameterized Algorithms Retreat of the University of Warsaw, PARUW 2020, held in Krynica-Zdrój in February 2020. This workshop was supported by a project that has received funding from the European Research Council (ERC) under the European Union's Horizon 2020 research and innovation programme under grant agreement No 714704 (PI: Marcin Pilipczuk).
}

\bibliographystyle{abbrv}
\bibliography{many-visits-tsp.bib}

\end{document}